\documentclass[final]{siamltex}

\usepackage{amsmath,wasysym,amssymb}
\usepackage[latin1]{inputenc}
\usepackage[active]{srcltx}
\usepackage{graphics}

\RequirePackage[colorlinks,citecolor=blue,urlcolor=blue]{hyperref}

\let\epsilon\varepsilon
\let\phi\varphi

\newtheorem{remark}{Remark}[section]
\newtheorem{example}[theorem]{Example}

\newcommand{\EE}{\mathbb{E}}

\newcommand{\R}{\mathbb{R}}
\renewcommand{\Re}[1]{\mathrm{Re}(#1)}

\newcommand{\cum}{{\textstyle \varint}}
\newcommand{\bigo}{\mathcal{O}}

\newcommand{\wro}[1]{w_{#1}}
\newcommand{\cof}[2]{d_{#1,#2}}
\newcommand{\ssum}[1]{a_{#1}}
\newcommand{\srec}[1]{\tilde{a}_{#1}}
\newcommand{\coeff}[2]{\alpha_{#1,#2}}

\providecommand{\abs}[1]{\lvert#1\rvert}

\title{Exact and asymptotic results for  insurance risk models with surplus-dependent premiums}

\author{Hansj\"org Albrecher\thanks{Department of Actuarial Science, Faculty of Business and Economics, University of
  Lausanne, Extranef Building, CH-1015 Lausanne, Switzerland, and Swiss Finance Institute ({\tt hansjoerg.albrecher@unil.ch}). Supported by the Swiss National Science Foundation Project~200021-124635/1.} \and
  Corina Constantinescu\thanks{Department of Actuarial Science, Faculty of Business and Economics, University of
  Lausanne, Extranef Building, CH-1015 Lausanne, Switzerland ({\tt corina.constantinescu@unil.ch}). Supported by the Swiss National Science Foundation Project~200021-124635/1.} \and
  Zbigniew Palmowski\thanks{Mathematical Institute, University of Wroclaw, Pl.~Grunwaldzki 2/4, 50-384 Wroclaw, Poland
 ({\tt zbigniew.palmowski@gmail.com})}\and
  Georg Regensburger\thanks{INRIA Saclay -- \^{I}le de France, Project DISCO, L2S, Sup\'{e}lec,
   91192 Gif-sur-Yvette Cedex, France ({\tt georg.regensburger@ricam.oeaw.ac.at}).
  Supported by the Austrian Science Fund (FWF): J3030-N18.} \and
  Markus Rosenkranz\thanks{School of Mathematics, Statistics \& Actuarial Science, University of Kent,
  Cornwallis Building, Canterbury, Kent CT2 7NF, United Kingdom ({\tt M.Rosenkranz@kent.ac.uk})}.}

\begin{document}

\maketitle

\begin{abstract}
In this paper we develop a symbolic technique to obtain asymptotic expressions for ruin probabilities and discounted penalty functions in renewal insurance risk models when the premium income depends on the present surplus of the insurance portfolio. The analysis is based on boundary problems for linear ordinary differential equations with variable coefficients. The algebraic structure of the Green's operators allows us to develop an intuitive way of tackling the
asymptotic behavior of the solutions, leading to exponential-type expansions  and Cram\'er-type asymptotics. Furthermore, we obtain closed-form solutions for more specific cases of premium functions in the compound Poisson risk model.
\end{abstract}



\pagestyle{myheadings}
\thispagestyle{plain}
\markboth{ALBRECHER, CONSTANTINESCU, PALMOWSKI, REGENSBURGER, ROSENKRANZ}{ON SURPLUS-DEPENDENT PREMIUM RISK MODELS}

\section{Introduction}
\label{sec:intro}

The study of level crossing events is a classical topic of risk theory and has turned out to be a fruitful area of applied mathematics, as (depending on the model assumptions) often subtle applications of tools from real and complex analysis, functional analysis, asymptotic analysis and also algebra are needed (see e.g. \cite{AsmAlb2010} for a recent survey).

In classical insurance risk theory, the collective renewal risk model describes
the amount of surplus $U(t)$ of an insurance portfolio at
time $t$ by \begin{equation}\label{eq1}U(t) = u +c\,t -\sum_{k=1}^{N(t)} X_k,\end{equation}
where $c$ represents a constant rate of premium inflow, $N(t)$ is a renewal process
that counts the number of claims incurred during the time interval
$(0,t]$ and $(X_k)_{k\geq 0}$ is a sequence of independent and identically distributed (i.i.d.) claim sizes with distribution function $F_X$ and density $f_X$ (also independent of the claim arrival process $N(t)$). Let $(\tau_k)_{k\geq 0}$ be the i.i.d. sequence of interclaim times. One of the crucial quantities to investigate in this context is the probability that at some point in time the surplus in the portfolio will not be sufficient to cover the claims, which is called the probability of ruin
\[
\psi(u) = P(T_u<\infty \mid U(0)=u),
\]
where $U(0)=u\geq 0$ is the initial capital in the portfolio and $$T_u=\inf \:\{t\geq 0:\,U(t)<0 \mid U(0)=u\}.$$
A related, more general quantity is the expected
discounted penalty function, which penalizes the ruin event for both
the deficit at ruin and the surplus before ruin,
\begin{align*}
  \Phi(u) &= \mathbb{E} \left(e^{-\delta T_u} \, w(U(T_u-),
    \abs{U(T_u)}) \, 1_{T_u <\infty}\,\vert \,U(0)=u\right),
\end{align*}
where $\delta\ge 0$ is a discount rate and the penalty $w(x,y)$ is a bivariate
function. ($\Phi(u)$ is often referred to as the Gerber-Shiu function, see \cite{GerberShiu1998}).

The classical collective risk model is based on the assumption of a constant premium rate
$c$. However, it is clear that it will often be more realistic to let premium amounts depend on the current
surplus level. In this case, the risk process (\ref{eq1}) is replaced by
$$U(t)=u+\int_0^t p(U(s))\;ds -\sum_{k=1}^{N(t)} X_k.$$
Hence, in between jumps (claims) the risk process moves
deterministically along the curve $\phi(u,t)$, which satisfies the
partial differential equation
\[
 \frac{\partial \phi}{\partial t} = p(u)\frac{\partial \phi}{\partial u};\quad \phi(u,0)=u.
\]
There are only a few situations for which exact expressions for $\psi(u)$ are known for surplus-dependent premiums. One such case is the Cram\'er-Lundberg risk model (where $N(t)$ is a homogeneous Poisson process with intensity $\lambda$) and the linear premium function $p(u)=c+\epsilon u$, which has the interpretation of an interest rate $\epsilon$ on the available surplus. In the case of exponential
claims, it was already shown by \cite{Seg1942} that the probability of ruin then has the
form
\begin{equation}\label{seg}
\psi(u) = \frac{\lambda \epsilon^{\lambda/\epsilon-1}}
{\mu^{\lambda/\epsilon}c^{\lambda/\epsilon}\exp(-\mu
  c/\epsilon)+\lambda \epsilon^{\lambda/\epsilon-1}\Gamma(\tfrac{\mu
    c}{\epsilon}, \tfrac{\lambda}{\epsilon})} \,
\Gamma\big(\tfrac{\mu(c+\epsilon
  u)}{\epsilon},\tfrac{\lambda}{\epsilon}\big),
\end{equation}
where $\Gamma(\eta, x) = \int_x^{\infty} t^{\eta -1}e^{-t}dt$ is the
incomplete gamma function (for extensions to finite-time ruin probabilities, see \cite{KnPet1,KnPet2} and \cite{AlTT}). In fact, for the Cram\'er-Lundberg risk model with exponential claims and general monotone
premium function $p(u)$, one has the explicit expression
\begin{equation} \label{dwa}
 \psi(u) = \int_u^\infty \frac{\gamma_0\lambda}{p(x)} \, \exp{\{\lambda q(x)-\mu x\}}\;dx,
\end{equation}
where $1/\gamma_0 \equiv 1+ \lambda \int_0^\infty p(x)^{-1} \, \exp
{\{\lambda q(x)-\mu x\}}\;dx$ and $q(x) \equiv \int_0^x
\frac{1}{p(y)}\,dy$ is assumed finite for $x>0$
(see \cite{Tichy84}).
Since for surplus-dependent premiums the probabilistic approach based on random equations does not work, and also the usual analytic methods lead to difficulties because the equations become too complex, it is a challenge to derive explicit solutions beyond the one given above. \\
In this paper we will employ a method based on boundary problems and Green's operators to derive closed-form solutions and asymptotic properties of $\psi(u)$ and $\Phi(u)$
under more general model assumptions.

%
%

%
%

For that purpose we will
employ the algebraic operator approach developed
in~\cite{Albetal2010}. However, since that approach was restricted to linear
ordinary differential equations (LODEs) with constant coefficients, we
will have to extend the theory to tackle the variable-coefficients equations that occur in the present context.

In Section~\ref{sec:boundary-problem} we derive the boundary problem for
the Gerber-Shiu function $\Phi(u)$ in a renewal risk model with claim and interclaim distributions having rational Laplace transform. For solving it, we employ a new
symbolic method, described in Section~\ref{sec:GOA}. This allows
to construct integral representations for the solution of
inhomogeneous LODEs with variable coefficients, for given initial values, under a stability condition. In
Section~\ref{sec:expansion} we derive a general asymptotic expansion for the discounted penalty function in the renewal model framework.
Subsequently, Section~\ref{sec:e1e1} is dedicated to the more specific case of compound Poisson risk models with exponential
claims, for which we have second-order LODEs. More specifically, in \ref{sec:exactgeneric} we derive exact solutions for a generic premium function $p(u)$. Further, in \ref{sec:exactpart}, we consider some interesting particular
cases of $p(u)$. In \ref{sec:asygeneric} we identify the necessary conditions a premium function should satisfy such that the asymptotic analysis is possible and the assumptions necessary for the asymptotic results in Section~\ref{sec:expansion} are validated. We will end by  giving concrete examples of such premium functions and their asymptotics.

Throughout the paper we will assume that $U(t)\to\infty$ a.s. This
assumption is satisfied for example when $p(u)>\EE X/\EE \tau +
\varsigma$ for some $\varsigma >0$ and sufficiently large $u$;
see e.g.~\cite{AsmAlb2010}.

\section{Deriving the boundary problem}
\label{sec:boundary-problem}

Assume that the distribution of the interclaim time of the renewal process $N(t)$ has rational Laplace transform. For simplicity of notation, we assume further that the rational Laplace transform has a constant numerator. Then its
density $f_{\tau}$ satisfies a LODE with constant coefficients
\begin{equation}
  \label{eq:tau}\mathcal{L}_{\tau}(\frac{d}{dt})f_{\tau}(t)=0\end{equation} and homogeneous
initial conditions $\smash{f_{\tau}^{(k)}(t)=0} \; (k=0,\ldots,
n-2)$, where
\begin{align*}
  &\mathcal{L}_{\tau}(x) =x^n +
  \alpha_{n-1}x^{n-1} +\cdots+\alpha_0.
\end{align*}
Using the method of~\cite{ConTho2011}, we can then derive an integro-differential equation for $\Phi(u)$
\begin{equation}
  \label{eq:ide}
  \mathcal{L}_{\tau}^*\left(p(u)\frac{d}{du}-\delta\right)
  \Phi(u) =\alpha_0 \left(\int_0^u \Phi(u-y) \, dF_X(y) + \omega(u)\right),
\end{equation}
where $\mathcal{L}_\tau^*$ is the adjoint operator of
$\mathcal{L}_\tau$ defined through
\[\mathcal{L}^*_{\tau}(x) =\mathcal{L}_{\tau}(-x) =(-x)^n +
  \alpha_{n-1}(-x)^{n-1} +\cdots+\alpha_0.\]

Assume now that the claim size distribution also has a rational Laplace
transform, so that its density $f_X$ satisfies another such LODE
\begin{equation}
  \label{eq:x}\mathcal{L}_{X}(\frac{d}{dy})f_{X}(y)=0\end{equation} with initial conditions
$\smash{f_{X}^{(k)}(x)=0} \: (k=0,\ldots, m-2)$, where
\begin{align*}
  &\mathcal{L}_{X}(x) =x^m + \beta_{m-1}x^{m-1} +\cdots+\beta_0.
\end{align*}
Then the
integro-differential equation \eqref{eq:ide} becomes a LODE with variable
coefficients of order $m+n$, namely
\begin{equation}
  \label{eq:ode}
  T \Phi(u) = g(u)
\end{equation}
with differential operator
\begin{equation}\label{opT}
 T= \mathcal{L}_{X}\left(\frac{d}{du}\right)\mathcal{L}_{\tau}^*\left(p(u)\frac{d}{du}-\delta\right) -\alpha_0 \beta_0
\end{equation}
 and right-hand side $$g(u)=\alpha_0 \,
  \mathcal{L}_{X}(\frac{d}{du}) \, \omega(u),$$ where $\omega(u) \equiv \int_u^{\infty}w(u, y-u) \, f_X(y) \, dy.$
For $\delta=0$ and $w=1$, Equation~\eqref{eq:ode}
reduces to the well-known equation for the probability of ruin.

The equations hold for sufficiently regular functions $p$. In the
special case $p(u)\equiv c$ one recovers the LODE with constant coefficients
whose characteristic polynomial is of degree $n+m$ and corresponds to
Lundberg's equation. It is known that, for $\delta>0$, this polynomial has $m$
solutions $\sigma_i$, with negative real part,  and $n$ solutions $\rho_i$, with positive real part; see for
example~\cite{LiGar2005}
and~\cite{LanWil2008}. In~\cite{Albetal2010}, we have derived
\begin{equation}
  \label{eq:ansatz}
  \Phi(u) = \gamma_1 e^{\sigma_1 u} +\cdots + \gamma_m  e^{\sigma_m u} + Gg(u),
\end{equation}
where the  $\gamma_i$ are determined by the initial conditions and
\begin{equation}
  \label{eq:oldgreen}
  Gg(u) \equiv \sum_{i=1}^{m}\sum_{j=1}^{n} c_{ij} \Big( \int_{0}^u
  e^{\sigma_i(u-\xi)} + \int_u^{\infty} e^{\rho_j(u-\xi)}
 - e^{\sigma_i u}\int_0^{\infty} e^{-\rho_j(\xi)}  \Big )\,g(\xi)\, d\xi
\end{equation}
defines the Green's operator for the inhomogeneous LODE~\eqref{eq:ode}
with homogeneous boundary conditions, where
\begin{align*}
  &c_{ij} = - \prod_{k=1,k\not=i}^m (\sigma_i - \sigma_k)^{-1}
  \prod_{k=1,k\not=j}^n (\rho_j-\rho_k )^{-1} \,(\rho_j-\sigma_i)^{-1}.
\end{align*}
The boundary conditions for \eqref{eq:ode} consist of the initial conditions
$\smash{\Phi^{(k)}(0)} \: (k = 0, \dots, m-1)$, determined from the integro-differential equation, and the stability
condition $\Phi(\infty) = 0$, provided by the model assumptions.

In analogy to the constant coefficients case, we assume the existence
of a fundamental system for equation \eqref{eq:ode} with~$m$ stable
solutions~$s_i(u)$ and $n$~unstable solutions~$r_j(u)$. Here a
solution~$f(u)$ is called stable if $f(u) \rightarrow 0$ and unstable
if $f(u) \rightarrow \infty$ as $u \rightarrow \infty$. We write $t_1,
\dots, t_{m+n}$ for the complete sequence of solutions $s_1, \ldots,
s_m, r_1, \dots, r_n$, and we assume furthermore that the successive
Wronskians $\wro{k} \equiv W[t_1, \dots, t_k]$, for $k = 1, \dots,
m+n$ are all nonzero on the half-line $\mathbb{R}^+ =[0, \infty)$. Under these
assumptions, the algebraic operator approach developed for the
constant coefficients case \cite{Albetal2010} will be extended to the
surplus-dependent premium case  in Section~\ref{sec:GOA}, and the general
solution of~\eqref{eq:ode} then has the form
\[
  \Phi(u) = \gamma_1 s_1(u)+\cdots + \gamma_m s_m(u) + Gg(u),
\]
where the $\gamma_i$ are determined by the initial values and $Gg(u)$ is
again the Green's operator for the inhomogeneous LODE~\eqref{eq:ode}
with homogeneous boundary conditions, but this time with non-constant
$p(u)$. As a consequence, the representation~\eqref{eq:oldgreen} is no
longer valid, and we will derive a new explicit expression that
generalizes it (Theorem~\ref{thm:gen-greensop}).


Let us complete this section with a remark about how to check
that the fundamental system has stable and unstable solutions. Roughly
speaking, this amounts to an asymptotic analysis of the solutions of
the homogeneous equation. According to \cite[Ch.5]{Fed1993}, one can
identify conditions on $p(u)$ that guarantee the existence of such a
fundamental system. These conditions specify the structure of the
coefficients, namely: either they converge (sufficiently fast) to
constants---in this case one speaks of \textit{almost constant
  coefficients}---or they diverge to infinity. The canonical form
of~\eqref{eq:ode} indicates of course that the former case applies for
our setting here. However, the speed of convergence of the
coefficients depends crucially on the premium function $p(u)$.
For
instance, we will show in Example~\ref{exx} that for $p(u) =
c\, e^{\epsilon/ u}$, the LODE with almost constant coefficients
converges to the LODE with constant coefficients given
in~\cite{Albetal2010}.

\section{Green's operator approach}
\label{sec:GOA}

In the previous section we have seen that the core task for computing
the Gerber-Shiu function $\Phi(u)$ is to determine the Green's operator
$G$ for the inhomogeneous LODE~\eqref{eq:ode} with homogeneous
boundary conditions consisting of the initial conditions
$\smash{\Phi^{(k)}(0) = 0} \: (k = 0, \dots, m-1)$ and the stability
condition $\Phi(\infty) = 0$. In this section we will present a symbolic
method that allows to construct $G$ for a generic LODE with variable
coefficients and homogeneous boundary conditions. In other words, we
consider boundary problems of the general type
\begin{align}
  \left\{
    \begin{aligned}
      \label{gen-bvp}
      & T \, \Phi(u) = g(u),\\
      & \Phi(0) = \Phi'(0) = \cdots =\Phi^{(m-1)}(0) = 0 \quad\text{and}\quad
      \Phi(\infty)=0,
    \end{aligned}
  \right.
\end{align}
where $T \equiv D^{m+n} + c_{m+n-1}(u) \, D^{m+n-1} + \dots + c_1(u)
\, D + c_0(u)$ is a linear differential operator with variable
coefficients (and leading coefficient normalized to unity) and $D \equiv
\tfrac{d}{du}$. Under the conditions described in
Section~\ref{sec:boundary-problem} the solution of~\eqref{gen-bvp} is
unique and depends linearly on the so-called forcing function
$g(u)$. Therefore the assignment $g \mapsto \Phi$ is a linear operator:
the Green's operator~$G$ of~\eqref{gen-bvp}.
The following fact follows immediately from the theory of differential equations.\\

\begin{theorem}\label{Th:main}
The Gerber-Shiu function equals
\begin{equation} \label{eq:genform}
\Phi(u) = \gamma_1 s_1(u)+\cdots+\gamma_m s_m(u)
 +Gg(u),
\end{equation}
where $G$ is the Green's operator for equation \eqref{gen-bvp},
and the constants $\gamma_i$ can be identified from the initial conditions.
\end{theorem}\\

For describing our new method of constructing an explicit
representation of~$G$, let us recall how this was achieved
in~\cite{Albetal2010} for the special case of constant coefficients
$c_i(u) \equiv c_i$. We will use the same notation as there, in particular
the basic operators $A = \cum_0^u$, $B = \cum_u^\infty$ and the
definite integral $F = A + B = \cum_0^\infty$. Employing the basic
operators, the crucial idea was to factor the Green's operator as
\begin{equation}
  \label{eq:oldgreenfact}
  G =(-1)^n A_{\sigma_1} \dots A_{\sigma_m} B_{\rho_1} \dots B_{\rho_n},
\end{equation}
where the factor operators are defined by $A_\sigma \equiv e^{\sigma
  x} A e^{-\sigma x}$ and $B_\rho \equiv e^{\rho x} B e^{-\rho x}$
with $\sigma_i$ and $\rho_j$ as described before. So the strategy was to decompose the problem and tackle the stable exponents with the basic operator~$A$,
the unstable ones with~$B$.

This idea can be carried over to the general case of~\eqref{gen-bvp}. Using
the results of \cite{RosenkranzRegensburger2008}, any Green's
operator can be fully broken down to basic operators if one can factor
the differential operator~$T$ into first-order factors. Having a
fundamental system $t_1, \dots, t_{m+n} = s_1, \dots, s_m, r_1, \dots,
r_n$ with successive Wronskians $\wro{k}(u) \ne 0 \; (k = 1, \dots,
m+n)$ for $u \in \mathbb{R}^+$, such a factorization of~$T$ can always be
achieved by well-known techniques described for example in Eqn.~(18)
of~\cite{Polya1922}; see also~\cite{Ristroph1972}
and~\cite{Zettl1974}. Using this factorization, we can break down $G$
in a way similar to~\eqref{eq:oldgreenfact} except that the
$A_{\sigma_i}$ must be replaced by more complicated operators based on
$A$ and $s_i$, similarly the $B_{\rho_j}$ by suitable operators
involving $B$ and $r_j$. We assume $m, n > 0$ throughout for avoiding
degenerate cases.\\

\begin{proposition}
  \label{greensop-fact}
  The Green's operator of~\eqref{gen-bvp} is given by $G = G_s G_r$
  where $G_s = A_{s_1} \cdots A_{s_m}$ and $G_r = (-1)^n B_{r_1}
  \cdots B_{r_n}$ with
  \begin{align*}
    A_{t_i} &= A_{s_i} = \tfrac{\wro{i}}{\wro{i-1}} \, A \,
    \tfrac{\wro{i-1}}{\wro{i}} && \quad\text{for $1 \le i \le m$},\\
    B_{t_j} &= B_{r_{j-m}} = \tfrac{\wro{j}}{\wro{j-1}} \, B \,
    \tfrac{\wro{j-1}}{\wro{j}} && \quad\text{for $m+1 \le j \le
      m+n$},
  \end{align*}
  setting $\wro{0} = 1$ for convenience.
\end{proposition}\\

\begin{proof}
  We employ the factorization $T = T_{r_n} \cdots T_{r_1} T_{s_m}
  \cdots T_{s_1}$, with the first-order operators given by
  \begin{align*}
    T_{t_i} &= \tfrac{\wro{i-1}}{\wro{i}} \, D \,
    \tfrac{\wro{i}}{\wro{i-1}} && \quad\text{for $1 \le i \le m$},\\
    T_{t_j} &= T_{r_{j-m}} = \tfrac{\wro{j-1}}{\wro{j}} \, D \,
    \tfrac{\wro{j}}{\wro{j-1}} && \quad\text{for $m+1 \le j \le
      m+n$}.
  \end{align*}
  It is then clear that $G = A_{s_1} \cdots A_{s_m} (-B_{r_1}) \cdots
  (-B_{r_n})$ is a right inverse of $T$ since both $A$ and $-B$ are
  right inverses of $D$. It remains to show that $\Phi = Gg$ satisfies
  the boundary conditions. Differentiating~$\Phi$ fewer than~$m$ times
  results in an expression whose summands all have the form $h \cdot
  (A \cdots g)$ for some functions $h$; evaluating any such summand
  yields $h(0) \cdot (\cum_0^0 \cdots g) = 0$, so the homogeneous
  initial conditions are indeed satisfied. For showing that the
  stability condition $\Phi(\infty) = 0$ is also fulfilled we write $\Phi =
  A_{s_1} \tilde{g}$ with $\tilde{g} \equiv A_{s_2} \cdots A_{s_m} G_r
  g$. Then $\Phi = s_1 A s_1^{-1} \tilde{g}$ and hence
  \begin{equation*}
   \Phi(\infty) = s_1(\infty) \, \cum_0^\infty s_1(u)^{-1} \tilde{g}(u) \,
    du = 0
  \end{equation*}
  because $s_1(\infty) = 0$ and the integral is assumed to
  converge.
\end{proof} \\

Note that we assume, in the above proof and henceforth, that all
forcing functions are chosen so that all infinite integrals have a
finite value (this will be the case in all the examples treated
here). This is also the reason why the $r_j$ are incorporated in $B$
operators rather than in $A$ operators as for the $s_i$. Since we want
to focus on the symbolic aspects here, we shall not elaborate these
points further.

Spelled out in detail, we can now write the Green's operator
of~\eqref{gen-bvp} in the factored form
\begin{align}
  \label{eq:goi-new}
  G &= \tfrac{\wro{1}}{\wro{0}} \, C_1 \, \tfrac{\wro{0} \,
    \wro{2}}{\wro{1}^2} \, C_2 \, \tfrac{\wro{1} \,
    \wro{3}}{\wro{2}^2} \, C_3 \, \cdots \, C_{m+n-1} \, \tfrac{\wro{m+n-2}
    \, \wro{n}}{\wro{m+n-1}^2} \, C_{m+n} \,
  \tfrac{\wro{m+n-1}}{\wro{m+n}},
\end{align}
where $C_i$ is $A$ for $1 \le i \le m$ and $-B$ for $m+1 \le i \le
m+n$. Although this brings us already some way towards a closed form
for $\Phi(u)$, we would like to collapse the $m+n$ integrals
of~\eqref{eq:goi-new} into a single integration, just as we did
in~\cite{Albetal2010}.

To start with, assume for a moment that we did not have any unstable
solutions so that the fundamental system is only $s_1, \dots, s_m$. In
that case we must dispense with the stability condition, imposing only
the homogeneous initial conditions in~\eqref{gen-bvp}. The Green's
operator consists only of $A$~operators, without any occurrence
of~$B$. In this simplified case, how can one collapse the $m$ integral
operators $C_1, \dots, C_m = A$ in~\eqref{eq:goi-new} by a linear
combination of single integrators (multiplication operators combined
with a single $A$)? The answer is given by the usual
variation-of-constants formula, which can be rewritten in our operator
notation as follows~\cite{RegensburgerRosenkranz2009a}.\\

\begin{proposition}
  \label{th:di}
  If $s_1, \dots, s_m$ is a fundamental system for the homogeneous
  equation $T \Phi=0$, the Green's operator of~\eqref{gen-bvp} is given
  by
 \begin{equation}
   \label{eq:goi-old}
   G_s = s_1 A \, \tfrac{\cof{m}{1}}{\wro{m}} +\cdots + s_m A \,
   \tfrac{\cof{m}{m}}{\wro{m}},
 \end{equation}
 where $\wro{m}$ is the Wronskian determinant of $s_1, \cdots, s_m$
 and $\cof{m}{i}$ results from $\wro{m}$ by replacing the $i$-th
 column by the $m$-th unit vector.
\end{proposition}\\

In other words, $\Phi = Gg$ is a particular solution of $T\Phi = g$, made
unique by imposing the initial conditions $\Phi(0) = \Phi'(0) = \cdots =
\Phi^{(m-1)}(0) = 0$. In our case, the stability condition $\Phi(\infty) =
0$ follows because $s_i(\infty) = 0$ for all $i=1, \dots, m$. But note
that~\eqref{eq:goi-old} is valid for any fundamental system $s_1,
\dots, s_m$ of $T$, yielding a particular solution for the initial
value problem (meaning~\eqref{gen-bvp} without the stability
condition).

Let us now turn to the general case, where the
fundamental system $t_1, \dots, t_{m+n}$ consists of $m\ge1$ stable
solutions $s_1, \dots, s_m$ and $n\ge1$ unstable solutions $r_1,
\dots, r_n$. In that case the Green's operator has a representation
analogous to~\eqref{eq:goi-old} except that we need $B$~operators in
addition to $A$~operators \emph{and} we have to include definite
integrals~$F$ for ``balancing'' the $B$ against the $A$~operators.\\

\begin{theorem}
  \label{thm:gen-greensop}
  Define the constants \begin{equation}
    \label{allpha}\coeff{i}{j} = \cof{i}{m+j}(0)/\wro{m+j-1}(0)\end{equation}
  for $j = 1, \ldots, n$ and $i = 1, \ldots, m+n$; the functions
  $\ssum{j} = \coeff{1}{j} \, s_1 + \cdots + \coeff{m}{j} \, s_m$ for
  $j = 1, \ldots, n$; and the functions $\srec{1}, \ldots, \srec{n}$
  by the recursion $\srec{1} = \ssum{1}$, $\srec{j} = \ssum{j} -
  \coeff{m+1}{j} \, \srec{1} - \cdots - \coeff{m+j-1}{j} \,
  \srec{j-1}$. Then the Green's operator of~\eqref{gen-bvp} is given
  by
  \begin{equation}
    \label{eq:gen-greensop}
    G = \sum_{i=1}^{m+n} t_i \, C_i \, \tfrac{\cof{i}{m+n}}{\wro{m+n}} -
    \sum_{j=1}^n \srec{j} \, F \, \tfrac{\cof{m+j}{m+n}}{\wro{m+n}},
  \end{equation}
  where $C_i$ is $A$ for $1 \le i \le m$ and $-B$ for $m+1 \le i \le m+n$.
\end{theorem} \\

The proof of this result is given in Appendix A. There is a more
explicit way of specifying the sequence of functions $\srec{1},
\ldots, \srec{n}$ occurring in Theorem~\ref{thm:gen-greensop}.\\

\begin{proposition}
  \label{det-func}
  The functions $\srec{j}$ in Theorem~\ref{thm:gen-greensop} can be
  computed by solving the system $T \srec{} = \ssum{}$, where $T$ is
  the lower triangular matrix with entries
  \begin{equation*}
    T_{jk} =
      \begin{cases}
        \alpha_{m+k}(j) & \text{for $j>k$,}\\
        1 & \text{for $j=k$,}\\
        0 & \text{otherwise,}
      \end{cases}
  \end{equation*}
  while $\srec{}$ and $\ssum{}$ are respectively columns with entries
  $\srec{1}, \ldots \srec{n}$ and $\ssum{1}, \ldots, \ssum{n}$. Hence
  we have explicitly $\srec{j} = \det{T_j}/\!\det{T}$, where $T_j$ is
  the matrix resulting from $T$ by replacing its $j$-th column by
  $\ssum{}$.
\end{proposition}\\

\begin{proof}
  We have $\coeff{m+1}{j} \, \srec{1} + \cdots + \coeff{m+j-1}{j} \,
  \srec{j-1} + \srec{j} = \ssum{j}$, for $j>1$, by the definition of
  the $\srec{j}$. But this is clearly the $j$-th row of the matrix $T
  \srec{}$, while the recursion base $\srec{1} = \ssum{1}$ provides
  the first row. The explicit formula is an application of Cramer's
  rule.
\end{proof} \\

In either form, the functions $\srec{1}, \dots, \srec{n}$ can be
readily computed from the given fundamental system $s_1, \dots, s_m,
r_1, \dots, r_n$, and the representation~\eqref{eq:gen-greensop}
provides a closed form for the Green's operator of~\eqref{gen-bvp}.

\section{Asymptotic results for the renewal risk model}
\label{sec:expansion}
In the sequel, we will write $k(u)\sim l(u)$ if $\lim_{u\to\infty}\frac{k(u)}{l(u)}=1$ for some functions $k$ and $l$. Assume that both the interclaim distribution and the
claim size distribution have rational Laplace transform, i.e. their densities satisfy the ODE \eqref{eq:tau}
and \eqref{eq:x}, respectively. Assume that the solutions of Equation
(\ref{eq:ode}) are of the form $t_i(u) \sim u^{\beta_i} \, e^{y_i u}$, i.e.
\begin{equation}\label{C1}
t_i(u)\sim \exp\left\{A\eta_i(u)\right\},\quad i=1,\ldots, n+m,
\end{equation}
with
\begin{equation}\label{C2}
\eta_i(u)\sim y_i+\frac{\beta_i}{u},\quad i=1,\ldots, n+m
\end{equation}
and
\begin{equation}\label{C3}
y_m<\ldots < y_1\leq 0<y_{m+1}<\ldots<y_{m+n}
\end{equation}
(so the $\eta_i$ are not asymptotically equivalent). \\

\begin{remark}\rm Note that for the premium functions $p(u)=c+\epsilon u$, $p(u) = c + \frac{1}{1+\epsilon u}$ and $p(u)= c \exp{\epsilon/u}$, the corresponding $t_i$ fulfill the conditions (\ref{C1})--(\ref{C3}). For  $m=n=1$, a more detailed analysis is presented in Section \ref{sec:asygeneric}.
\end{remark} \\

%

Define the constants $h_k=\gamma_k-\sum_{j=1}^n \alpha_{jk} \, F \, \frac{\cof{m+j}{m+n}}{\wro{m+n}}$
with $\gamma_k$ appearing in \eqref{eq:genform} and $\alpha_{jk}$ as defined in \eqref{allpha}.
%
 For a permutation $\varphi$ on $\{1, \dots, m+n\}$ we define $$ \pi_i=\frac{\sum_{\varphi(i)=n+m}(-1)^{{\rm sgn}\varphi}\prod_{k\neq i}y_k^{\varphi(k)}}{\sum_{\varphi} (-1)^{{\rm sgn}\varphi}
 \prod_{k=1}^{n+m}y_k^{\varphi(k)}},$$  where ${\rm sgn}{\varphi}$ denotes the parity of $\varphi$.\\

\begin{theorem}\label{th:main2}
If $g(u)\sim e^{-\nu u}$ for $\nu>-y_1$, then under (\ref{C1})--(\ref{C3}) the asymptotic expansion
\begin{equation}\label{expansion}
\Phi(u)\approx\sum_{i=1}^{m+1}\vartheta_i(u) \end{equation}
holds, with $\vartheta_i(u)=h_is_i(u)\; (i=1,\ldots,m)$ and $$\vartheta_{m+1}(u)\sim \sum_{i=1}^{m+n} \frac{\pi_i}{y_i+\nu} \,g(u).$$
\end{theorem}\\

This is equivalent to saying  that $\lim_{u\to\infty}\frac{\Phi(u)-\sum_{i=1}^k\vartheta_i(u)}{\vartheta_{k+1}(u)}= 1,$ for $k=1,\ldots, m.$

\begin{proof}
Note that by (\ref{C1})--(\ref{C2}),
$t^{(k)}_i(u)\sim y_i^{k}e^{A\eta_i(u)}.$
Using \eqref{eq:gen-greensop} and the Leibniz formula for the determinant,
after some calculations one gets that expansion (\ref{expansion})
with $\vartheta_k(u)=l_kt_k(u)$ and
$$\vartheta_{m+1}(u) = \sum_{i=1}^{m+n} \pi_i t_i(u) C_i \frac{g}{t_i}(u).$$
Using l'H\^{o}pital's rule completes the proof.
\end{proof}

\section{Compound Poisson risk process with exponential claims}
\label{sec:e1e1}

Let us now focus on the case of a compound Poisson model
(exponential interclaim times with mean $1/\lambda$) with exponential claim sizes with mean $\mu$ and a generic
premium function $p(u)$. The differential equation in~\eqref{eq:ode}
has order two in this case, so we expect to have one stable solution $s$ and
one unstable solution $r$. In fact, here we can relax the notion
of an unstable solution, allowing any function where
\begin{equation*}
  r(\infty) = \lim_{u \rightarrow \infty} r(u)
\end{equation*}
exists and is different from zero (so the limit does not necessarily have to be infinity). The
reason for this extension is that the basic argument for the ansatz
\begin{equation*}
  \label{eq:lincomb}
  \Phi(u) = \gamma_s s(u) + \gamma_r r(u)
\end{equation*}
carries over: Every solution of~\eqref{eq:ode} must be of the
form~\eqref{eq:lincomb} since $r(u), s(u)$ forms a fundamental
system. But then the stability condition $\Phi(\infty) = 0$ can only be
satisfied if $\gamma_r = 0$ because we require $s(\infty) = 0$. This is why
the form~\eqref{eq:genform} is still justified in the special case
$n=1$ with $\gamma_1 = \gamma_s$. But note that this argument fails when
there are more than two unstable solutions since they can cancel out
unless we take some further precautions (e.g.\@ requiring them to be
of the same sign).

\subsection{Closed-form solutions for generic premium}
\label{sec:exactgeneric}

For a discount factor $\delta > 0$,
the expected discounted penalty functions satisfies the second order
LODE
\begin{equation*}
  \left(D+\mu\right) \left(-p(u)D+\delta +\lambda \right)  \Phi(u)
  -\lambda \mu \, \Phi(u) = \lambda (D +\mu) \, \omega(u).
\end{equation*}
Expanding the operators, the equation is equivalent with
\begin{equation*}
 \left(- p(u) D^2-(\mu \, p(u) +p^\prime(u)-\lambda -\delta) \, D +\delta
   \mu \right) \Phi(u) = \lambda (D +\mu) \, \omega(u).
\end{equation*}
Assuming that $p(u) \ne 0$ for all $u\geq 0$, this is further
equivalent to
\begin{equation}
  \label{eq:odee1e1}
  \left(D^2+\left(\mu +\frac{p^\prime(u)}{p(u)}-\frac{\lambda
        +\delta}{p(u)}\right)D - \frac{\delta \mu}{p(u)} \right) \Phi(u) = g(u),
\end{equation}
with $g(u)= -\frac{\lambda}{p(u)} (D +\mu) \, \omega(u)$. Furthermore,
we assume that $p(u)$ is chosen in such a way that the associated
homogeneous solution has a fundamental system $s, r$ with one stable
solution $s$ and one unstable solution $r$ with Wronskian $w=\wro{2} = s
r' - s' r$ nonzero on $\mathbb{R}^+$. Then the Green's operator for the
boundary problem for the Gerber-Shiu function~$\Phi$ is given by
Theorem~\ref{thm:gen-greensop} with $s_1= s$ and $r_1 = r$, namely
\begin{equation} \label{eq:greenop} Gg(u)= \Big(-s(u) \int_0^u
  \tfrac{r(v)}{w(v)} - r(u) \int_u^{\infty}
  \tfrac{s(v)}{w(v)} +\tfrac{r(0)}{s(0)}s(u)\int_0^{\infty}
  \tfrac{ s(v)}{w(v) }\Big) g(v)\, dv.
\end{equation}
For calculating the full expression
\begin{equation} \label{eq:m}
\Phi(u)= \gamma \, s(u) +Gg(u)
\end{equation}
 we
have to determine the constant $\gamma$. Evaluating the
integro-differential equation \eqref{eq:ide} at zero, one obtains
\begin{equation*}
 -c \, \Phi'(0)+(\lambda+\delta) \, \Phi(0) = \lambda \, \omega(0)
\end{equation*}
and therefore
\begin{equation} \label{eq:const}
  \gamma =\frac{\lambda \, \omega(0) + c \, (Gg)'(0) }{(\lambda
    +\delta) \, s(0) - c \, s'(0)} = \frac{\lambda \, \omega(0) +
    c \, \frac{r(0) s'(0) - r'(0) s(0)}{s(0)} \int_0^{\infty}
    \frac{s(v)}{w(v)} \, g(v) \, dv}{(\lambda +\delta) \, s(0) -
    c \, s'(0)}
\end{equation}
for the required constant.
 For $\delta=0$, the LODE~\eqref{eq:odee1e1} is of
 first order in $\Phi'$, and its associated homogeneous equation has an
 unstable solution $r(u)=1$ and a stable solution
 \begin{equation} \label{eq:asm} s(u) = \int_u^{\infty} \exp\Big(-\mu v
   + \lambda \int_0^v \tfrac{dy}{p(y)}\Big) \, \frac{ dv}{p(v) }
 \end{equation}
 (cf.~\cite{AsmAlb2010}).
 For the fundamental system $s,r$, the Wronskian is just $w=\wro{2} =
 -s'$, and the Green's operator~\eqref{eq:greenop} specializes to
 \begin{equation*}
   Gg(u)= \Big(s(u) \int_0^u \tfrac{1}{s'(v)}  +   \int_u^{\infty}
   \tfrac{s(v)}{s'(v)} -\tfrac{s(u)}{s(0)}\int_0^{\infty} \tfrac{
     s(v)}{s'(v) }\Big) \, g(v)\, dv
 \end{equation*}
 while the constant in $\Phi(u)= \gamma \, s(u) +Gg(u)$ is now given by
 \begin{equation*}
   \gamma = \frac{\lambda \omega(0) - p(0) \tfrac{s'(0)}{s(0)}
     \int_0^{\infty} \tfrac{s(v)}{s'(v)}\,g(v)\,dv}{\lambda s(0) - p(0) s'(0)}.
 \end{equation*}
 Thus the Gerber-Shiu function can be written generically as
 \begin{align*}
   \Phi(u) &= \frac{\lambda \omega(0) - p(0) \tfrac{s'(0)}{s(0)}
     \int_0^{\infty} \tfrac{s(v)}{s'(v)}\,g(v)\,dv}{\lambda s(0) - p(0)
     s'(0)} \, s(u)\\[1ex]
   &\qquad + \Big(s(u) \int_0^u \tfrac{1}{s'(v)} + \int_u^{\infty}
   \tfrac{s(v)}{s'(v)} -\tfrac{s(u)}{s(0)}\int_0^{\infty} \tfrac{
     s(v)}{s'(v) } \Big) \, g(v)\, dv
 \end{align*}
 in terms of $s(u)$.\\

\begin{remark}\rm  For $\delta=0$ and $w=1$, one has $g=0$ and $\psi(u)= \gamma \,
s(u)$, recovering \eqref{dwa} for the ruin probability.
\end{remark}

\subsection{Closed-form solutions for some particular premium structures}
\label{sec:exactpart}

{\bf A) Linear premium:} As discussed in Section \ref{sec:intro}, the linear function $p(u)=c+\epsilon u$ can be interpreted as
describing investments of the surplus into a bond with a fixed interest rate
$\epsilon >0$; see for example~\cite{Seg1942}.
For $\delta>0$ and $p(u)=c+\epsilon u$, we can compute a fundamental
system for the second-order LODE
\begin{equation*}
  \Big(D^2+\left(\mu +\tfrac{\epsilon}{c+ \epsilon u}-\tfrac{\lambda
        +\delta}{c+ \epsilon u}\right)D -\tfrac{\delta \mu}{c+\epsilon
      u} \Big) \, \Phi(u) = - \tfrac{\lambda}{c+ \epsilon u} \, (D
    +\mu) \, \omega(u)
\end{equation*}
in the form
\begin{equation}
  \label{m1m2}
  \begin{aligned}
    s(u) &= U\big(\tfrac{\delta}{\epsilon}+1, \,
    \tfrac{\lambda+\delta}{\epsilon} +1, \, \mu u + \tfrac{\mu
      c}{\epsilon} \big)\, ( \epsilon u +
    c)^{\tfrac{\lambda+\delta}{\epsilon}} \, \exp(-\mu u),\\
    r(u) &= M\big(\tfrac{\delta}{\epsilon}+1, \,
    \tfrac{\lambda+\delta}{\epsilon} +1, \, \mu u + \tfrac{\mu
      c}{\epsilon} \big)\, ( \epsilon u +
    c)^{\tfrac{\lambda+\delta}{\epsilon}} \, \exp(-\mu u),
  \end{aligned}
\end{equation}
where $M$ and $U$ denote the usual Kummer functions as
in~\cite[\S13.1]{AbrSte1964}. For $u \rightarrow \infty$, the
estimate in~\S13.1.8 yields
\begin{equation} \label{eq:s}
 s(u) = K_1 \, ( \epsilon u+ c)^{\lambda/\epsilon-1} \, \exp(-\mu u) \, \Big(1+ \bigo(\tfrac{1}{\epsilon u + c}) \Big) \rightarrow 0,
\end{equation}
while the estimate in~\S13.1.4 yields
\begin{equation} \label{eq:r}
  r(u) = K_2 \, (\epsilon u+c)^{\delta/\epsilon} \, \Big(1+
  \bigo(\tfrac{1}{\epsilon u +c}) \Big)=K_2 \, (\epsilon
  u+c)^{\delta/\epsilon}  \,+\bigo((\epsilon u+c)^{\delta/\epsilon -
    1}) \rightarrow \infty,
\end{equation}
where $K_1$ and $K_2$ are some constants. Hence $s$ is indeed a stable
and $r$ an unstable solution. Using~\S13.1.4
one derives the Wronskian
\begin{equation*}
  \wro{2} = \tfrac{\Gamma(\frac{\lambda+\delta}{\epsilon})}{\Gamma(\frac{\delta}{\epsilon})}\tfrac{\epsilon(\lambda +\delta)}{\delta} \big(\tfrac{\epsilon}{\mu}\big)^{(\lambda+\delta)/\epsilon}  (u\epsilon + c)^{(\lambda+\delta)/\epsilon-1} \exp(-\mu u+ \tfrac{\mu c}{\epsilon}).
\end{equation*}
Substituting these expressions in \eqref{eq:greenop}, we end up with
\begin{align*}
  & Gg(u)=
  \tfrac{\Gamma(\delta/\epsilon +
    1)}{\Gamma((\delta + \lambda)/(1+\epsilon))}
  \tfrac{1}{\epsilon}
  \left(\tfrac{\mu}{\epsilon}\right)^{(\lambda+\delta)/\epsilon}
  (\epsilon u +c)^{(\lambda+\delta)/\epsilon} \exp(-\mu u
  -\tfrac{\mu c}{\epsilon})\times \\
  &\quad \Big(-U(u) \int_0^u M(v) - M(u) \int_u^{\infty} U(v) +
  \tfrac{M(0)}{U(0)} \, U(u) \int_0^{\infty} U(v)\Big)\, g(v)\,dv,
\end{align*}
where $U(u)$ and $M(u)$ are  Kummer functions appearing on the right hand side
of ~\eqref{m1m2}. This jointly with \eqref{eq:const} is sufficient to determine the discounted penalty function in \eqref{eq:m}.

{\bf B) Exponential premium: } In general, an exponential premium function leads to intractable
results. However, for
\begin{equation*}
  p(u)=c(1+ e^{-u})
\end{equation*}
the probability of ruin can be worked out from the expression in Section \ref{sec:exactgeneric}:
\begin{equation*}
  \psi(u)= -\frac{(1+\frac{\lambda}{c})F(\frac{\lambda}{c},
    \mu; 1+\frac{\lambda}{c}; e^u+1)(\frac{1}{2} e^u
    +\frac{1}{2})^{\frac{\lambda}{c}}}{2\mu F(1+\mu,
    1+\frac{\lambda}{c}; 2+\frac{\lambda}{c}; 2)}
\end{equation*}
where $F(a,b;c;z)={_2}F_1(a,b;c;z)$ stands for the hypergeometric
function, see e.g. \cite{AbrSte1964}.

{\bf C) Rational premium: } For a basic rational premium like
\begin{equation*}
  p(u) = c + \frac{1}{1+u},
\end{equation*}
the exact symbolic form for the probability of ruin can be
computed up to quadratures, namely
\begin{equation*}
  \psi(u)=\frac{\lambda (c+1)^{\lambda/c^2} \int_u^\infty e^{-u (c
      \mu-\lambda)/c} (c+c u+1)^{-(\lambda+c^2)/c^2}
    (1+u)du}{1+\lambda (c+1)^{\lambda/c^2} \int_0^\infty e^{-u (c
      \mu-\lambda)/c} (c+c u+1)^{-(\lambda+c^2)/c^2} (1+u)du}.
\end{equation*}

{\bf D) Quadratic premium: } For the quadratic function $p(u)= c + u^2$, the probability of ruin
can be determined as
\begin{equation*}
\psi(u)= \frac{  \lambda \int_u^{\infty} e^{-(-\lambda
    \arctan(x/\sqrt{c})+\mu x \sqrt{c})/\sqrt{c}}/(c+x^2)\,
  dx}{1+\lambda \int_0^{\infty}e^{-(-\lambda \arctan(x/\sqrt{c})+\mu x
    \sqrt{c})/\sqrt{c}}/(c+x^2)\, dx}.
\end{equation*}
%

\subsection{Asymptotic results for generic premium}
\label{sec:asygeneric}

Assume the LODE
\begin{equation}
  \label{eq:acc}
  \Phi^{(n)} + c_{n-1}(u) \, \Phi^{(n-1)} + \cdots + c_0(u) \, \Phi =0
\end{equation}
has complex coefficients $c_i(u)$ continuous on $\mathbb{R}^+$ and define its
characteristic equation as
\begin{equation} \label{chareq:acc}
 y^n+ c_{n-1}(u) \, y^{n-1} + \cdots + c_0(u) =0.
\end{equation}
Then the asymptotic behavior of solutions of~\eqref{eq:acc} for
$u\rightarrow \infty$ essentially depends on the behavior of the roots
$y_1(u), \ldots, y_n(u)$ of~\eqref{chareq:acc} as $u \rightarrow
\infty$ (see e.g. ~\cite[\S5.3.1, p.~250]{Fed1993}), which will be exploited below.

\subsubsection{Probability of ruin}

When $\delta=0$ and $w=1$, the expected discounted penalty is the
probability of ruin. For this quantity we have the following
asymptotic estimate (we use the convention $p(\infty)=\lim_{u\to\infty}p(u)$):\\

\begin{proposition} \label{ruinprobas}
  \begin{enumerate}
  \item If $p(\infty) =c$, where $c$ is constant, then
    $$\psi(u) \sim\frac{\mu}{\lambda}\gamma \exp\Big(-\mu u + \lambda \int_0^u \tfrac{dw}{p(w)}\Big), \quad u\rightarrow \infty. $$
  \item  If $p(\infty) =\infty$, then
    $$\psi(u) \sim\frac{\mu}{\lambda}\gamma \frac{1}{p(u)}\exp\Big(-\mu u +
    \lambda \int_0^u \tfrac{dw}{p(w)}\Big), \quad u\rightarrow
    \infty. $$
  \end{enumerate}
\end{proposition}

\begin{proof}
Integration by parts in \eqref{eq:asm} gives
\begin{equation*}
 s(u) = \frac{\mu}{\lambda}\int_u^{\infty}  \exp\Big(-\mu v + \lambda \int_0^v \tfrac{dw}{p(w)}\Big) \,  dv - \frac{1}{\lambda}\exp\Big(-\mu u + \lambda \int_0^u \tfrac{dw}{p(w)}\Big).
\end{equation*}
with $s(0)=\frac{\mu}{\lambda} \hat{h}(\mu) -\frac{1}{\lambda},$ $s'(0) =-\frac{1}{ p(0)},$
where $\hat{h}$ denotes the Laplace transform of $h(u)= \exp(\lambda \int_0^u \tfrac{dw}{p(w)})$.
Letting  $f(u) = \frac{1}{\lambda}\int_u^{\infty}  \exp\Big(-\mu v + \lambda \int_0^v \tfrac{dw}{p(w)}\Big) \,  dv,$ one gets
\begin{eqnarray*}\psi(u) = \gamma s(u) =\mu f(u)-f'(u) = \mathcal{L}^*_X(\frac{d}{du}) f(u),
\end{eqnarray*}
with $\gamma=\frac{\lambda}{\lambda s(0) -p(0) s'(0)}$.
Note that
$f''(u)=\mu f'(u)\left(1+\frac{1}{p(u)}\right). $
To prove the first part of the proposition we need to show that  $\psi(u)= -\mu \frac{1}{1+c}f'(u)(1+o(1)). $
According to our previous observation,
$ \frac{\mu f'(u) -f''(u)}{-\mu f''(u)} =\frac{1}{1+p(u)},$ which completes the proof using l'H\^{o}pital rule.
The second part can be proved similarly.
\end{proof}

\subsubsection{Expected discounted penalty}

We consider two cases of premium functions:
\begin{enumerate}
\item[P1.] the premium function behaves like a constant at infinity \begin{equation}\label{pierwszyprzypadek}
p(\infty)=c,\qquad p^\prime(u)= O\left(\frac{1}{u^2}\right);
\end{equation}
or
\item[P2.] the premium function explodes at infinity, $p(\infty)=\infty$ as
\begin{equation}\label{drugiprzypadek}
p(u)=c+\sum_{i=1}^l\epsilon_i u^i,\qquad c>0.
\end{equation}
\end{enumerate}
The first case is satisfied by the rational and exponential premium functions.
The second case is satisfied by the linear and quadratic premium functions (see Section~\ref{sec:exactpart}).
Consider first the homogenous equation, \eqref{eq:odee1e1} with $g=0$,  i.e. equation \eqref{eq:ode} with $T$ given in \eqref{opT}
with
$$
c_0(u)=-\frac{\delta \mu}{p(u)},\quad c_1(u)=\mu +\frac{p^\prime(u)}{p(u)}-\frac{\lambda +\delta}{p(u)}.$$
After tedious calculations one can check that in the case (\ref{pierwszyprzypadek}) we have
$$\lim_{u\to\infty}\frac{c_1(u)}{\sqrt{c_0(u)}}<\infty$$
so that Conditions 1) and 2) of \cite[p. 252]{Fed1993} are satisfied.
In the second case (\ref{drugiprzypadek}) we have
$$\lim_{u\to\infty}\frac{c_0(u)}{c_1^2(u)}=0$$
and Conditions 1) and 2') of \cite[p. 254]{Fed1993} hold. From ~\cite[p.252]{Fed1993} we hence know that for both cases \eqref{pierwszyprzypadek} and \eqref{drugiprzypadek} the asymptotic behavior of the solution of \eqref{eq:odee1e1} is

\begin{equation}\label{ti}
t_i(u) \sim \exp\left\{\int_{0}^u (\varrho_i(t) +\varrho^{(1)}_i(t) )dt \right\},\quad i=1,2, \end{equation}
where
\begin{equation*}
\varrho_1=\frac{-\left(\mu +\frac{p'(u)}{p(u)}-\frac{\lambda +\delta}{p(u)}\right) -
\sqrt{\left(\mu +\frac{p'(u)}{p(u)}-\frac{\lambda +\delta}{p(u)}\right)^2+4\tfrac{\delta \mu}{p(u)}}}{2}
\end{equation*}
and
\begin{equation*}
\varrho_2=\frac{-\left(\mu +\frac{p'(u)}{p(u)}-\frac{\lambda +\delta}{p(u)}\right) +
\sqrt{\left(\mu +\frac{p'(u)}{p(u)}-\frac{\lambda +\delta}{p(u)}\right)^2+4\tfrac{\delta \mu}{p(u)}}}{2}
\end{equation*}
are the negative and positive solution, respectively, of the characteristic equation
\begin{equation} \label{eq:char}
 x^2 +\left(\mu +\frac{p'(u)}{p(u)}-\frac{\lambda +\delta}{p(u)}\right) x -\frac{\delta \mu}{p(u)} =0.
\end{equation}
Here $\varrho_1^{(1)}$ and $\varrho_2^{(1))}$  are defined by
 $$\varrho^{(1)}_i(u) = -\frac{\varrho'_i(u)}{2\varrho_i(u) + \left(\mu +\frac{p'(u)}{p(u)}-\frac{\lambda +\delta}{p(u)}\right)}
 =\frac{\varrho'_i(u)}{\sqrt{\left(\mu +\frac{p'(u)}{p(u)}-\frac{\lambda +\delta}{p(u)}\right)^2+4\tfrac{\delta \mu}{p(u)}}}.$$
\begin{remark}\rm Note that if the premium function $p(u)$ satisfies conditions \eqref{pierwszyprzypadek} and \eqref{drugiprzypadek}, the solutions $t_i(u)$ will be of the asymptotic form \eqref{C1}, where $\eta_i=\varrho_i+\varrho_i^{(1)}$.
\end{remark}
In order to complete the asymptotic analysis of $\Phi(u)$ for large $u$, recall that the Gerber-Shiu function $\Phi$ is given by
$\Phi(u)= \gamma s(u)+Gg(u),$
for a normalizing constant $\gamma$. \\

\begin{theorem}\label{th:asy}
 Under the assumptions \eqref{pierwszyprzypadek} and \eqref{drugiprzypadek} regarding the premium function, the asymptotics of the Gerber-Shiu function are described by
$$\Phi(u) \sim h_1 s(u) +K_1 g(u), $$ with the exception
$$\Phi(u) \sim h_1 s(u) +K_2 u \, g(u) $$ for the case \eqref{drugiprzypadek} with $l=1$. Here $h_1=\gamma-\int_0^{\infty} \tfrac{ s(v)}{s'(v) }g(v)\, dv.$
\end{theorem}\\

\begin{remark}
 Moreover, for the particular examples considered here, the structure of $s$ (and $r$) is indeed that of the form \eqref{C1}--\eqref{C3} that we had to impose as a condition in the more general framework.
\end{remark}\\

\begin{proof}
First note that for $\delta=0,$ $Gg(u)=0,$ and thus $\Phi(u)$ has the same behavior as the probability of ruin $\psi(u)=\gamma s(u)$.
Evaluating the expression (\ref{ti}) at $\delta=0$,
%
\begin{eqnarray*}
 s(u) &\sim&    e^{-\mu u +\lambda \int_0^u \frac{dv}{p(v)}}\,
 \left(\mu p(u) + p'(u)-\lambda \right)^{-1}
\end{eqnarray*}
leads to the classical result regarding the probability of ruin \eqref{dwa}.
For $\delta \neq 0$, one needs the asymptotic behavior of $Gg(u),$ which based on \eqref{eq:greenop} can be reduced to analyzing the asymptotic behavior of
$$q(u)= -s(u) \int_0^u
  \tfrac{r(v)}{w(v)}\,g(v) dv - r(u) \int_u^{\infty}
  \tfrac{s(v)}{w(v)}\,g(v) dv, $$
since the term $s(u)\int_0^{\infty} \tfrac{ s(v)}{s'(v) }g(v)\, dv$ behaves as $s(u)$ at infinity. After rewriting
\begin{equation*}
q(u)= -\frac{\int_0^u
  \tfrac{r(v)}{w(v)}\,g(v) dv}{\frac{1}{s(u)}}-\frac{\int_u^{\infty}
  \tfrac{s(v)}{w(v)}\,g(v) dv}{\frac{1}{r(u)}}
\end{equation*}
and expanding the Wronskian, one can apply l'H\^{o}pital rule and see that as $u\rightarrow \infty$ (after some algebra)
\begin{equation} \label{eq:asyh}
q(u)\sim \frac{\frac{1}{\tfrac{r'(u)}{r(u)}-\tfrac{s'(u)}{s(u)}}g(u)}{\tfrac{s'(u)}{s(u)}}-\frac{\frac{1}{\tfrac{s'(u)}{s(u)}-\tfrac{r'(u)}{r(u)}}g(u)}{\tfrac{r'(u)}{r(u)}}.
\end{equation}
%
Using Fedoryuk's asymptotic expressions (\ref{ti}) one more time, one can perform the analysis along the two cases introduced here.
It easy to check that in the first case, P1, 
we have
\begin{equation}\label{firstcaseas}
s(u)\sim \exp\left\{- k_1 u\right\}, \quad r(u)\sim \exp\left\{- k_2 u\right\},
\end{equation}
where $k_1=- \frac{\left(\mu -\frac{\lambda +\delta}{c}\right) +\sqrt{\left(\mu -\frac{\lambda +\delta}{c}\right)^2+4\tfrac{\delta \mu}{c}}}{2}$ and $k_2 =- \frac{\left(\mu -\frac{\lambda +\delta}{c}\right) -\sqrt{\left(\mu -\frac{\lambda +\delta}{c}\right)^2+4\tfrac{\delta \mu}{c}}}{2}.$
Thus, as  $\lim_{u\rightarrow \infty}\frac{s'(u)}{s(u)} =k_1$ and $\lim_{u\rightarrow \infty}\frac{r'(u)}{r(u)} =k_2$,  then
\begin{equation} \label{firstcaseh}
 h(u)\sim K_1\,g(u), \quad u\rightarrow \infty,
\end{equation}
where $K_1=\frac{k_1+k_2}{k_1k_2 (k_2-k_1)}= \frac{\mu -\frac{\lambda +\delta}{c}}{\tfrac{\delta \mu}{c} \sqrt{\left(\mu -\frac{\lambda +\delta}{c}\right)^2+4\tfrac{\delta \mu}{c}}}.$
The second case, P2, 
is more complex, producing more intriguing asymptotics. One can show that in this case $$s(u)\sim h_1 \, u^{\beta}\, e^{-\mu u},$$ with $\beta\in \mathbb{R}$. Note that for $l=1$, $\epsilon_1=\epsilon$ one recovers the asymptotics \eqref{eq:s}. One can also check that
$$\lim_{u\rightarrow \infty}\frac{s'(u)}{s(u)} =-\mu,$$ whereas
\begin{equation*}
 \frac{r'(u)}{r(u)} \sim \frac{1}{u} \quad \mbox{for}\;\,l=1,\quad \mbox{and}\quad  r(u) \sim 1  \quad \mbox{for}\;\, l>1,
\end{equation*}
%
producing respectively
\begin{equation} \label{secondcaseh}
q(u) \sim u g(u) \quad \mbox{and}\quad  q(u)\sim g(u).
\end{equation}
\end{proof}
\begin{example}\rm
 Consider a compound Poisson risk model with premium functions described by assumptions (P1) or (P2). Let the penalty  be a function of the surplus only, $w(x,y)=e^{-\nu x}$. Since we are in the exponential claims scenario,
$$g(u) = \lambda \mu (D+\mu) \int_u^{\infty} w(u-y) e^{-\mu y}\, dy=\lambda \nu e^{-(\nu +\mu)u}.$$
Thus, for a linear premium function, $$\Phi(u) =\gamma s(u)+ Gg(u) \Phi(u)\sim h_1 u^{\beta} e^{-\mu u}+\lambda \nu u e^{-(\nu +\mu)u},$$ with $\beta=\lambda/\epsilon -1$,  whereas for all the other premium funtions in the class considered here,
$$\Phi(u)\sim h_1 u^{\beta} e^{-\mu u}+\lambda \nu e^{-(\nu +\mu)u}, \quad u\rightarrow \infty,$$
with $\beta\in \mathbb{R}$.
\end{example}

\begin{example}\label{exx}\rm
 When $p(u)= c \exp{\epsilon/u} $, one has a differential  equation with \textit{almost constant coefficients},
\begin{eqnarray} \nonumber
  &&\left(D^2+(\mu -\frac{\epsilon}{u^2}-\frac{\lambda+\delta}{c}
\exp{-\epsilon/u})D-\frac{\delta \mu}{c} \exp{-\epsilon/u}\right)  \Phi(u)\\ \label{eq:odeexp}
&&=-\frac{\lambda}{c} \exp{-\epsilon/u}  (D +\mu)\omega(u).
\end{eqnarray}
This is an  equation of form \eqref{eq:acc}, with coefficients satisfying
\begin{equation}
 c_k(u)= \alpha_k +a_k(u), \quad k=1,2
\end{equation}
with $\alpha_k$ constant and $ \int_1^\infty |a_k(u)| du < \infty$. Here $a_1(u) = -\frac{\epsilon}{u^2} +\frac{\lambda+\delta}{c}a(u)$ and $a_0(u)= \frac{\delta \mu}{c} a(u)$, with  $a(u)= \sum_{k=1}^\infty \frac{(-1)^n \epsilon^n}{u^n n!}$ thus $\int_1^\infty |a(u)|du <\infty,$ and similarly for $a_0$ and $a_1$. From \cite[Th. 8.1, p. 92]{CoddLev} (see also Problem 32, p.105 there) we can hence conclude that the homogeneous equation has a fundamental system with asymptotics
$$s(u)=e^{\sigma u}\left(1+o(1)\right) \quad \mbox{and} \quad r(u)=e^{\rho u}\left(1+o(1)\right),  $$
where $\sigma$ and $\rho$ are solutions of the equation
\begin{equation*}
 x^2 +\left(\mu -\frac{\lambda+\delta}{c}\right)x -\frac{\delta \mu}{c}  =0,
\end{equation*}
with $\Re{\sigma}<0$  and $\Re{\rho}>0$.
Note that these solutions coincide with the one of the constant premium case. Consequently, one has the same asymptotic behavior as when the premium rate is constant.
\end{example}

\section{Conclusion}
\label{sec:concl}
We have provided a symbolic method and a conceptual framework for studying boundary value problems with variable coefficients as they appear in modelling the surplus level in a portfolio of insurance contracts in classical risk theory. The approach presented  allows a detailed analysis of the asymptotic behavior of the solutions of these equations under a set of conditions. For the specific case of the compound Poisson risk model, these conditions were made more explicit in terms of conditions on the form of $p(u)$. Moreover, several new closed-form solutions were established within this framework.

\Appendix
\section{Proof of Theorem~\ref{thm:gen-greensop}}
\label{sec:BPHL}

The proof of
Theorem~\ref{thm:gen-greensop} hinges on the following technical
lemma on Wronskian determinants.\\

\begin{lemma}
  \label{wronskian-formula}
  We have $\big( \tfrac{\cof{i}{k+1}}{\wro{k}} \big)' = - \cof{i}{k}
  \, \tfrac{\wro{k+1}}{\wro{k}^2}\,$ for $1 \le i \le k < m+n$.
\end{lemma} \\

\begin{proof}
  We have to show $\cof{i}{k+1} \, \wro{k}' - \cof{i}{k+1}' \, \wro{k}
  = \cof{i}{k} \, \wro{k+1}$. We note that all expressions in this
  formula are certain minors of the Wronskian matrix $W$ for $t_1,
  \dots, t_{m+n}$. So let us write $W^{i_1, \ldots, i_l}_{j_1, \ldots,
    j_l}$ for the minor of $W$ resulting from deleting the columns
  indexed $i_1, \dots, i_l$ and the rows indexed $j_1, \dots,
  j_l$. Then we have $\wro{k} = W^{k+1}_{k+1}$, $\cof{i}{k+1} =
  (-1)^{i+k+1} W^i_{k+1}$, $\cof{i}{k} = (-1)^{i+k}
  W^{i,k+1}_{k,k+1}$, with the derivatives $\cof{i}{k+1}' =
  (-1)^{i+k+1} W^i_k$ and $\wro{k}' = W^{k+1}_k$. For the latter, we
  use the fact that a Wronskian determinant can be differentiated if
  one replaces the last row by its derivative; see for example
~\cite[p.~118]{Ince1926}. Multiplying by $(-1)^{i+k}$, it remains to show
  $W^i_k W^{k+1}_{k+1} - W^i_{k+1} W^{k+1}_k = W^{i,k+1}_{k,k+1} \cdot
  \det{W}$. But this is a classical determinant formula of Sylvester;
  see for example~\cite[p.~1571]{Salem2008} or Eqn.~(4.49'')
  in~\cite{Groebner1966}.
\end{proof}

The preceding lemma is the key tool for removing the nested integrals
in~\eqref{eq:gen-greensop}. For seeing this, note that it can be read
backwards as giving the integral of $\cof{i}{k} \,
\wro{k+1}/\wro{k}^2$. In conjunction with certain operator identities
taken from~\cite{Rosenkranz2005}, this allows us to collapse
expressions of the form $A \cdots A$ or $B \cdots B$ or, at the
interface of the two blocks, $A \cdots B$. \\

\begin{proof}[Proof of Theorem~\ref{thm:gen-greensop}]
  Note that the case $n=0$ reduces to Proposition~\ref{th:di}, so we
  may assume $n>0$ in the sequel. We know from
  Proposition~\ref{greensop-fact} that
  \begin{equation*}
    G = (-1)^n A_{s_1} \cdots A_{s_m} B_{r_1} \cdots B_{r_n},
  \end{equation*}
  using the notation employed there. Based on this factorization, we
  prove~\eqref{eq:gen-greensop} by induction on $n$.
  In the base case $n=1$, applying Proposition~\ref{th:di} again
  yields
  \begin{align*}
    G = & \;G_s (-B_{r_1}) = \Big( \sum_{i=1}^m s_i \, A \,
    \tfrac{\cof{m}{i}}{\wro{m}} \Big) \tfrac{\wro{m+1}}{\wro{m}} \,
    (-B) \, \tfrac{\wro{m}}{\wro{m+1}}\\
    &= \sum_{i=1}^m s_i \, A \, \Big( - \cof{m}{i} \,
    \tfrac{\wro{m+1}}{\wro{m}^2} \Big) \, B \,
    \tfrac{\wro{m}}{\wro{m+1}},
  \end{align*}
  and Lemma~\ref{wronskian-formula} gives $(\cof{m+1}{i}/\wro{m})'$ for
  the expression in parentheses. Now we employ the identity $A f B = A
  \, (\cum_0^u f) + (\cum_0^u f) \, B$, for arbitrary functions~$f$,
  from~\cite{Rosenkranz2005}. Substituting the expression in
  parentheses for~$f$, this gives
  \begin{equation*}
    A \, \tfrac{\cof{m+1}{i}}{\wro{m}} + \tfrac{\cof{m+1}{i}}{\wro{m}} \,
    B - \coeff{1}{i} \, F,
  \end{equation*}
  so we end up with
  \begin{equation*}
    G = \sum_{i=1}^m \Big( s_i \, A \, \tfrac{\cof{m+1}{i}}{\wro{m+1}} + s_i \,
    \tfrac{\cof{m+1}{i}}{\wro{m}} \, B \tfrac{\wro{m}}{\wro{m+1}} -
    \coeff{1}{i} \, s_i \, F \tfrac{\wro{m}}{\wro{m+1}} \Big).
  \end{equation*}
  In the middle, we factor out $\sum_i s_i \, \cof{m+1}{i}$, which equals
  $-r_1 \, \cof{m+1}{m+1}$ as one sees by replacing the last row
  in~$\wro{m+1}$ by the first and then expanding along that last
  row. But~$\cof{m+1}{m+1} = \wro{m}$, so the middle sum simplifies to
  $r_1 \, (-B) \, \cof{m+1}{m+1}/\wro{m+1}$ and may thus be
  incorporated into the first sum. In the third sum of the above
  expression, we factor out $\sum_i \coeff{1}{i} \, s_i = \ssum{1}$. Thus we
  obtain finally
  \begin{equation*}
    G = \sum_{i=1}^{m+1} t_i \, C_i \, \tfrac{\cof{m+1}{i}}{\wro{m+1}} -
    \srec{1} \, F \, \tfrac{\cof{m+1}{m+1}}{\wro{m+1}},
  \end{equation*}
  which is the desired formula~\eqref{eq:gen-greensop} for $n=1$.
 Now assume Equation~\eqref{eq:gen-greensop} for $n$; we prove it for
 $n+1$. Using the induction hypothesis we obtain
  \begin{align*}
    G &= (-1)^n A_{s_1} \cdots A_{s_m} B_{r_1} \cdots B_{r_n}
    (-B_{r_{n+1}})\\
    &= \Big( \sum_{i=1}^{m+n} t_i \, C_i \,
    \tfrac{\cof{i}{m+n}}{\wro{m+n}} - \sum_{j=1}^n \srec{j} \, F \,
    \tfrac{\cof{m+j}{m+n}}{\wro{m+n}} \Big) \,
    \tfrac{\wro{m+n+1}}{\wro{m+n}} \, (-B) \,
    \tfrac{\wro{m+n}}{\wro{m+n+1}}\\
    &= \sum_{i=1}^{m+n} t_i \, C_i \, \big(- \cof{i}{m+n} \,
    \tfrac{\wro{m+n+1}}{\wro{m+n}^2} \big) \, B \,
    \tfrac{\wro{m+n}}{\wro{m+n+1}}\\
    & \qquad\qquad - \sum_{j=1}^n \srec{j} \, F \,
    \big(- \cof{m+j}{m+n} \tfrac{\wro{m+n+1}}{\wro{m+n}^2} \big) \,
    B \, \tfrac{\wro{m+n}}{\wro{m+n+1}}.
  \end{align*}
  As before, we see that Lemma~\ref{wronskian-formula}, with $n+1$ in
  place of $n$, can be applied to the expressions in the two
  parentheses, yielding $(\cof{m+n+1}{i}/\wro{m+n})'$ for the former and
  $(\cof{m+n+1}{m+j}/\wro{m+n})'$ for the latter. In addition to the
  identity for $A f B$ used for the base case, we need now also the
  related identities $B f B = (\cum_u^\infty f) \, B - B \,
  (\cum_u^\infty f)$ and $F f B = F \, (\cum_0^u f)$ also to be found
  in~\cite{Rosenkranz2005}. When we substitute for $f$, these
  identities take on the form
  \begin{align*}
    A f B &= A \, \tfrac{\cof{m+n+1}{i}}{\wro{m+n}} +
    \tfrac{\cof{m+n+1}{i}}{\wro{m+n}}  \, B - \coeff{n+1}{i} \, F,\\
    B f B &= B \, \tfrac{\cof{m+n+1}{i}}{\wro{m+n}} -
    \tfrac{\cof{m+n+1}{i}}{\wro{m+n}} \, B
  \end{align*}
  for the first expression and
  \begin{equation*}
    FfB = F \, \cof{m+n+1}{m+j}/\wro{m+n} - \coeff{n+1}{m+j} \, F
  \end{equation*}
  for the second. We split the first sum above into the two sums
  \begin{align*}
    &\sum_{i=1}^m s_i \, \big( A \, \tfrac{\cof{m+n+1}{i}}{\wro{m+n}} +
    \tfrac{\cof{m+n+1}{i}}{\wro{m+n}} \, B - \coeff{n+1}{i} \, F \big)
    \, \tfrac{\wro{m+n}}{\wro{m+n+1}},\\
    &\sum_{j=1}^n r_j \big( \tfrac{\cof{m+n+1}{m+j}}{\wro{m+n}} \, B -
     B \, \tfrac{\cof{m+n+1}{m+j}}{\wro{m+n}} \big) \,
     \tfrac{\wro{m+n}}{\wro{m+n+1}}.
  \end{align*}
  In the lower-range sum
  \begin{align*}
    & \sum_{i=1}^m s_i \, A \, \tfrac{\cof{m+n+1}{i}}{\wro{m+n+1}} +
    \Big( \sum_{i=1}^m s_i \, \cof{m+n+1}{i} \Big)/\wro{m+n} \, B \,
    \tfrac{\wro{m+n}}{\wro{m+n+1}}\\
    & \qquad\qquad - \Big( \sum_{i=1}^m \coeff{n+1}{i} \, s_i \Big) \,
    F \, \tfrac{\wro{m+n}}{\wro{m+n+1}}
  \end{align*}
  we can apply the same determinant expansion as before to obtain
  \begin{align*}
    & \sum_{i=1}^m s_i \, A \, \tfrac{\cof{m+n+1}{i}}{\wro{m+n+1}}
    + \Big(-r_{n+1} - \sum_{j=1}^n r_j \,
    \tfrac{\cof{m+n+1}{m+j}}{\wro{m+n}} \Big) \, B \,
    \tfrac{\wro{m+n}}{\wro{m+n+1}}\\
    & \qquad\qquad - \ssum{n+1} \, F \,
    \tfrac{\wro{m+n}}{\wro{m+n+1}}\\
    & = \sum_{i=1}^m t_i \, C_i \, \tfrac{\cof{m+n+1}{i}}{\wro{m+n+1}} +
    t_{m+n+1} \, C_{m+n+1} \, \tfrac{\cof{m+n+1}{m+n+1}}{\wro{m+n+1}}\\
    & \qquad\qquad -
    \sum_{j=1}^n r_j \, \tfrac{\cof{m+n+1}{m+j}}{\wro{m+n}} \, B \,
    \tfrac{\wro{m+n}}{\wro{m+n+1}} -
    \ssum{n+1} \, F \, \tfrac{\wro{m+n}}{\wro{m+n+1}};
  \end{align*}
  in the upper-range sum we get
  \begin{equation*}
    \sum_{j=1}^n r_j \tfrac{\cof{m+n+1}{m+j}}{\wro{m+n}} \, B \,
    \tfrac{\wro{m+n}}{\wro{m+n+1}} + \sum_{j=1}^n
     t_{m+j} \, C_{m+j} \, \tfrac{\cof{m+n+1}{m+j}}{\wro{m+n+1}}.
  \end{equation*}
  Combining the lower-range with the upper-range sum, the first sum
  within the latter cancels with the second sum within the former,
  yielding
  \begin{equation*}
    \sum_{i=1}^{m+n+1} t_i \, C_i \, \tfrac{\cof{m+n+1}{i}}{\wro{m+n+1}} -
    \ssum{n+1} \, F \, \tfrac{\wro{m+n}}{\wro{m+n+1}}.
  \end{equation*}
  Now let us tackle the second sum in the above expression for $G$,
  namely
  \begin{align*}
    &- \sum_{j=1}^n \srec{j} \, F \, \big(- \cof{m+j}{m+n}
    \tfrac{\wro{m+n+1}}{\wro{m+n}^2} \big) \,
    B \, \tfrac{\wro{m+n}}{\wro{m+n+1}}\\
    &= \sum_{j=1}^n \srec{j} \, \big( \coeff{n+1}{m+j} \, F - F \,
    \tfrac{\cof{m+n+1}{m+j}}{\wro{m+n}} \big) \,
    \tfrac{\wro{m+n}}{\wro{m+n+1}}\\
    &= \Big( \sum_{j=1}^n \coeff{n+1}{m+j} \, \srec{j} \Big) \, F \,
    \tfrac{\wro{m+n}}{\wro{m+n+1}} - \sum_{j=1}^n \srec{j} \, F
    \, \tfrac{\cof{m+n+1}{m+j}}{\wro{m+n+1}},
  \end{align*}
  where the expression in parentheses is $\ssum{n+1} - \srec{n+1}$
  by the definition of the $\srec{j}$. Altogether we obtain now
  \begin{align*}
    G &= \sum_{i=1}^{m+n+1} t_i \, C_i \,
    \tfrac{\cof{m+n+1}{i}}{\wro{m+n+1}}  - \sum_{j=1}^n \srec{j} \, F
    \, \tfrac{\cof{m+n+1}{m+j}}{\wro{m+n+1}}\\
    & \qquad\quad - \Big( \ssum{n+1} \, F \,
    \tfrac{\wro{m+n}}{\wro{m+n+1}} -
    (\ssum{n+1} - \srec{n+1}) \, F \,
    \tfrac{\wro{m+n}}{\wro{m+n+1}} \Big)\\
    &= \sum_{i=1}^{m+n+1} t_i \, C_i \,
    \tfrac{\cof{m+n+1}{i}}{\wro{m+n+1}}  - \sum_{j=1}^{n+1} \srec{j} \, F
    \, \tfrac{\cof{m+n+1}{m+j}}{\wro{m+n+1}},
  \end{align*}
  which is indeed~\eqref{eq:gen-greensop} with $n+1$ in place of $n$.
\end{proof} \\

In concluding this Appendix, let us also mention that
Theorem~\ref{thm:gen-greensop} is also valid if $B$ is taken to be the
operator $\cum_x^b$ with finite $b \in \R$ rather than $b =
\infty$. The reason is that the operator identities
from~\cite{Rosenkranz2005} are also valid in this case (and were
actually set up for this case in the first place).

\subsection*{Acknowledgments}
The authors would like to thank the Mathematical Institute of
Wroclaw  and the University of Lausanne for accommodating several research visits.

\bibliographystyle{siam}

\end{document}